\documentclass[11pt,journal,draftcls,letterpaper,onecolumn]{IEEEtran}

\usepackage{amsmath, amsthm, amssymb,graphicx,subfigure,algorithm,algorithmic}
\usepackage{times, cite,multirow}
\newtheorem{theorem}{Theorem}

\begin{document}
\title{On the Universality of Sequential Slotted Amplify and Forward Strategy in Cooperative Communications}
\author{Haishi~Ning,~\IEEEmembership{Student Member,~IEEE,} Cong~Ling,~\IEEEmembership{Member,~IEEE,} and~Kin~K.~Leung,~\IEEEmembership{Fellow,~IEEE}%
\thanks{This work was published in part in the IEEE International Conference on Communications 2009 (ICC 2009), Dresden, Germany, June 2009.}%
\thanks{The authors are with Department of Electrical and Electronic Engineering, Imperial College London, London SW7 2AZ, UK (e-mail: haishi.ning06@imperial.ac.uk; c.ling@imperial.ac.uk; kin.leung@imperial.ac.uk).}}
\maketitle
\begin{abstract}
While cooperative communication has many benefits and is expected to play an important role in future wireless networks, many challenges are still unsolved. Previous research has developed different relaying strategies for cooperative multiple access channels (CMA), cooperative multiple relay channels (CMR) and cooperative broadcast channels (CBC). However, there lacks a unifying strategy that is universally optimal for these three classical channel models. Sequential slotted amplify and forward (SSAF) strategy was previously proposed to achieve the optimal diversity and multiplexing tradeoff (DMT) for CMR. In this paper, the use of SSAF strategy is extended to CBC and CMA, and its optimality for both of them is shown. For CBC, a CBC-SSAF strategy is proposed which can asymptotically achieve the DMT upper bound when the number of cooperative users is large. For CMA, a CMA-SSAF strategy is proposed which even can exactly achieve the DMT upper bound with any number of cooperative users. In this way, SSAF strategy is shown to be universally optimal for all these three classical channel models and has great potential to provide universal optimality for wireless cooperative networks.
\end{abstract}
\section{Introduction}
\label{introduction}
Cooperative communication, which can significantly improve the efficiency and robustness of wireless communication systems, has been a hot research topic recently that attracted a lot of research attention \cite{sen1,sen2,janani,nabar,bletsas,mitran,prasad,kramer,stefanov,liu}. Many relaying strategies for cooperative networks with various topologies have been proposed with the trend of allowing the source nodes to transmit as much as possible and using nonorthogonal signal subspace as much as possible \cite{azarian}. These cooperative strategies are often compared using the diversity and multiplexing tradeoff (DMT) \cite{tse}, which is a fundamental measure that characterizes throughput and error performance simultaneously.

For cooperative multiple access channels (CMA) with multiple cooperative source nodes, a single destination node but without any dedicated relay node, nonorthogonal amplify and forward (NAF) strategy, which protects half of sources' signal by allowing a source node to relay its previously received signal simultaneously with another source node transmitting a new independent message, has been shown to be optimal \cite{azarian}. For cooperative multiple relay channels (CMR) with a single source node, a single destination node and multiple dedicated relay nodes, sequential slotted amplify and forward (SSAF) strategy, which permits each relay node to forward its previously received signal in a specific assigned time slot sequentially, has been shown to be asymptotically optimal \cite{shengyang}. The best known relaying strategy for cooperative broadcast channels (CBC) with a single source node, multiple cooperative destination nodes but without any dedicated relay node is CBC-DDF, which is based on the dynamic decode and forward (DDF) strategy that allows a destination node to help others only after it has successfully decoded the desired information. However, CBC-DDF has been shown to be suboptimal in the high multiplexing gain regime \cite{azarian}.

The optimality of a strategy is defined as the ratio between the DMT upper bound for any strategy and the achievable DMT lower bound for this specific strategy tends to $1$ as the total number of nodes in the network increases. In a wireless cooperative network, one natural DMT upper bound for a destination node is the multiple-input single-output (MISO) DMT upper bound, which is obtained by viewing all other nodes in the network as being physically connected through a genie with perfect knowledge about the source information and willing to devote all their resource to help the destination node to receive the desired information better. For CMA with $N$ source nodes, CMR with $N-1$ relay nodes and CBC with $N$ destination nodes, the MISO DMT upper bound is
\begin{eqnarray}
d(r)=N(1-r)^+.
\end{eqnarray}

The suboptimality of the CBC-DDF strategy motivated us to develop a better transmission strategy for CBC which can approach the DMT upper bound. More importantly, in practical wireless networks, nearly every transmission is a broadcast process, nearly every reception is a multiple access process and a large number of wireless nodes need to function as relays. Thus, we are particularly interested in finding a transmission strategy that is universally optimal for CMR, CBC and CMA. This is an important problem because a wireless node can be a source, a destination or a relay from time to time, and we need a unifying strategy that can provide optimality no matter the specific role of the wireless node. In this paper, we first briefly review the SSAF strategy and some preliminaries. Then, we extend the use of the powerful SSAF strategy, which was previously proved to be optimal for CMR \cite{shengyang}, to CBC and CMA, and show it asymptotically achieves the DMT upper bounds for both of them. In this way, we show SSAF strategy is universally optimal for CMR, CBC and CMA.
\subsection{SSAF strategy}
SSAF strategy was previously proved to be asymptotically optimal for CMR in some special cases \cite{shengyang}. The original SSAF strategy there allows the source to transmit during all the time slots. Then, from the beginning of the second time slot, there is one and only one relay forwarding a scaled version of what it received in the previous slot. To simplify the analysis, the consecutive relays are assumed to be isolated, i.e., the channel gain between two consecutively used relays is approximated as $0$, and the SSAF strategy is shown to achieve the genie-aided DMT upper bound for any SAF scheme [13, Eqn. 2], which itself asymptotically achieves the MISO DMT upper bound as the number of slots increases.

In general, SSAF is a class of transmission strategies with the following properties:
\begin{enumerate}
\item One cooperation frame is composed of multiple time slots.
\item The source nodes are assigned one or more time slots to transmit their independent information in a sequential manner.
\item The relay nodes are assigned one or more time slots to amplify and forward linear combination of their previously received signal in a sequential manner.
\item The destination nodes receive signal from the source and the relay nodes in every time slot when they do not transmit.
\end{enumerate}

SSAF strategy allows at least one source node to transmit a new message in every time slot in order to maximize the multiplexing gain. Moreover, it encourages to use nonorthogonal signal subspace in order to maximize the diversity gain and compensate for the effect of the half-duplex constraint.

The relay nodes using SSAF strategy has low processing complexity because they only need to scale and retransmit their previously received signal. Decoding is only needed at the destination nodes and not until the end of one cooperation frame. The scheduling complexity is low because every node operates sequentially and equally. Thus, the scheduling problem is a direct extension of that of a single source node, a single relay node or a single destination node problem.

One can design a SSAF variant which works for a particular network model. Different SSAF variants may assign different time slots to one node to transmit its own information and forward its previously received signal. The original SSAF strategy in \cite{shengyang} was proposed for CMR with dedicated relay nodes. Since we want to design SSAF variants for CBC and CMA without dedicated relays, we will later propose two SSAF variants which choose the source and destination nodes to act as relay nodes in two different sequential manners. However, general rules and associated attractive properties still hold for every SSAF variant we propose.
\subsection{Preliminaries}
In this paper, we use $s_m,1\leqslant{m}\leqslant{M}$ to denote the $M$ source nodes, $t_n,1\leqslant{n}\leqslant{N}$ to denote the $N$ destination nodes and $e_k,1\leqslant{k}\leqslant{K}$ to denote the $K$ relay nodes. We always omit the subscript of $s_m$, $t_n$ or $e_k$ when there is only one source node, one destination node or one relay node if no confusion can be raised.

Every node is constrained by average energy $E$. All source nodes transmit independent information at the same rate $R$. We use $x_{s_m,l}$ and $x_{e_k,l}$ to represent the transmit signal from the $m$th source node and the $k$th relay node at the $l$th time slot, and we use $y_{e_k,l}$ and $y_{t_n,l}$ to represent the receive signal at the $k$th relay node and the $n$th destination node at the $l$th time slot. $h_{s_m,e_k}$, $h_{s_m,t_n}$ and $h_{e_k,t_n}$ are used to denote the channel gain between the $m$th source node and the $k$th relay node, the channel gain between the $m$th source node and the $n$th destination node and the channel gain between the $k$th relay node and the $n$th destination node. We assume the existing physical links are all quasi-static flat Rayleigh-fading, which means the channel gains are constant during each cooperation frame but change independently between different frames.

The mathematical tools we use are mainly from \cite{tse,laneman1,laneman,azarian,shengyang,telatar}. The transmit SNR of a physical link is defined as $\rho=\frac{E}{\sigma^2}$ where $E$ is the average signal energy at the transmitter and $\sigma^2$ is the noise variance at the receiver. We say $b$ is the exponential order of $f(\rho)$ if $\lim_{\rho\rightarrow\infty}\frac{\log(f(\rho))}{\log(\rho)}=b$ and denote $f(\rho)$ as $f(\rho)\dot{=}\rho^b$. $\dot{\leqslant}$ and $\dot{\geqslant}$ are similarly defined.
Consider a coding scheme as a family of codes $\{C(\rho)\}$ with data rate $R(\rho)$ bits per channel use (BPCU) and average maximum-likelihood (ML) error probability $P_E(\rho)$. The multiplexing gain $r$ and the diversity gain $d$ are defined as \cite{tse}
\begin{eqnarray}
r=\lim_{\rho\rightarrow\infty}\frac{R(\rho)}{\log(\rho)},\quad d=-\lim_{\rho\rightarrow\infty}\frac{\log(P_E(\rho))}{\log(\rho)}.
\label{eqn:basicdef}
\end{eqnarray}
\section{CBC-SSAF strategy}
\begin{figure}
    \centering
    \includegraphics[width=6cm]{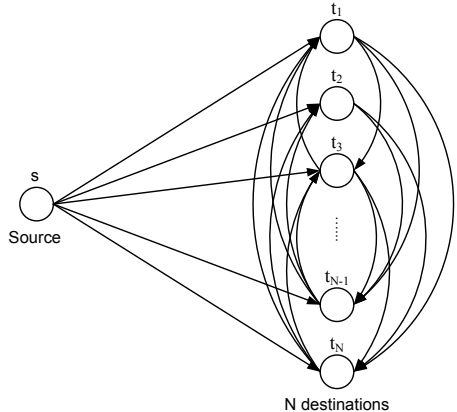}
    \caption{Cooperative broadcast channel with a single source node and $N$ cooperative destination nodes.}
    \label{fig:cbc}
\end{figure}
In this section, we consider the cooperative broadcasting problem as shown in Fig. \ref{fig:cbc}, where a single source node $s$ wants to broadcast its messages (may contain common and individual messages) to all the $N$ destination nodes $t_1,t_2,...,t_N$. We propose a near-optimal relaying strategy CBC-SSAF for CBC, based on the class of SSAF strategies. CBC-SSAF allows each destination node to act in turn as a relay node and forward its previously received signal to other destination nodes. This is similar to a single-cell downlink phase with cooperative users.

Before giving details about the proposed CBC-SSAF strategy, we firstly point out that due to the half-duplex constraint, CBC-SSAF is not a full multiplexing gain strategy. This is because in one cooperation frame, each destination node has to spend one time slot not listening in order to act as a relay node. Thus, our strategy cannot achieve the full multiplexing gain $1$ with finite time slots. When the number of time slots goes to infinity, this loss is negligible and thus CBC-SSAF is asymptotically optimal.

Secondly, it is easy to see that the highest probability of error (which corresponds to a DMT lower bound) occurs when the source node $s$ only transmits common messages and all the destination nodes want to decode every transmitted message from the source \cite{azarian}. Thus, we only investigate this worst case scenario which corresponds to a DMT lower bound.

Now, we detail the proposed CBC-SSAF strategy as follows:
\begin{enumerate}
\item One cooperation frame contains $N+1$ different time slots.
\item $s$ keeps transmitting a new message $x_{s,l}$, for $l=1,2,...,N+1$, in each corresponding time slot.
\item In the first time slot, select the destination node $t_1$ as the first relay node such that the channel gain between $s$ and $t_1$ is larger than the channel gains between $s$ and other destination nodes. The first relay node receives the strongest signal from the source node.
\item In the $l$th time slot, for $2\leqslant{l\leqslant{N}}$, select the destination node $t_l$ as the $l$th relay node such that the channel gain between the $(l-1)$th relay node and the $l$th relay node is as small as possible. Thus, the $l$th relay node keeps silent in the $l$th time slot to receive a cleaner signal from $s$ with smallest possible interference from the previous relay node.
\item In the $(l+1)$th time slot, for $1\leqslant{l\leqslant{N}}$, the selected $l$th relay node broadcasts a normalized version of its received signal in the previous time slot.
\end{enumerate}

We dynamically order the destination nodes to act as relay nodes before the start of each cooperation frame, such that the next effective relay (NER) is chosen as the worst destination for the current effective relay (CER). The purpose of such relay pre-ordering is to separate consecutively used relays as much as possible by making the channel gain between two consecutively used relays as small as possible. In other words, the relay pre-ordering here can be viewed as a method to mimic the isolated relays assumption in \cite{shengyang} which significantly simplifies the analysis.

A simple method to implement the relay pre-ordering operations can be done using Algorithm \ref{algorithm:cbc}. The timing structure after the relay pre-ordering operations is shown in Fig. \ref{fig:timing}, where solid boxes denote transmitted signal and dashed boxes denote received signal. Note that the relay pre-ordering algorithm here needs to estimate $N$ source-to-destination channel gains and $N(N-1)$ destination-to-destination channel gains. Thus, a central scheduler which collects all these channel state information may generate overwhelming overhead. Fortunately, due to the wireless broadcast nature, the scheduling task can be done distributively at the source and destination nodes as follows:

Before the start of each cooperation frame, let each destination sequentially broadcasts a short \textit{probe frame}. These operations consume $N$ \textit{probe frame} time slots. After reception of these \textit{probe frames}, the source node can estimate the $N$ source-to-destination channel gains and choose the first relay node based on step $4$ in Algorithm \ref{algorithm:cbc}. Because of the wireless broadcast nature, a destination node can at the same time receive $N-1$ \textit{probe frames} and estimate its local $N-1$ destination-to-destination channel gains. If we assume the wireless reciprocity property holds, a destination can use this information to choose its NER based on step 6 in Algorithm \ref{algorithm:cbc}. Then, each destination sends back a short \textit{feedback frame} to the source with only its chosen NER's unique ID embedded in it. These operations consume $N$ \textit{feedback frame} time slots. After decoding these \textit{feedback frames}, the source node can construct a linked list locally and the relay pre-ordering operations can be now completely. In each \textit{data frame} time slot, the source only needs to embed a unique ID in its signature. Each destination node extracts this ID and if it matches its own, it acknowledges it should function as a relay in the next \textit{data frame} time slot. The extra payloads of these relay pre-ordering operations are $N$ \textit{probe frame} time slots and $N$ \textit{feedback frame} time slots, which can be well assumed to be much shorter than the $N+1$ \textit{data frame} time slots. Moreover, considering the underlying block fading assumption, the cost of the scheduling algorithm is negligible.
\begin{algorithm}
\caption{Relay pre-ordering algorithm for CBC-SSAF strategy.}
\label{algorithm:cbc}
\begin{algorithmic}[1]
\STATE Before each cooperation frame, set $i=1$ and $\mathcal{T}=\{t_1,t_2,...,t_N\}$
\WHILE {$i\neq{N+1}$}
\IF {$i=1$}
    \STATE Choose $t_m\in\mathcal{T}$ such that $h_{s,t_m}\geqslant{}h_{s,t_n}$, $\forall{}t_n\in{\mathcal{T}}$, $n\neq{m}$
\ELSE
    \STATE Choose $t_m\in\mathcal{T}$ such that $h_{t_{i-1},t_m}\leqslant{}h_{t_{i-1},t_n}$, $\forall{}t_n\in{\mathcal{T}}$, $n\neq{m}$
\ENDIF
\STATE Swap the indexes for $t_m$ and $t_{i}$
\STATE Delete $t_{i}$ from $\mathcal{T}$
\STATE{$i=i+1$}
\ENDWHILE
\end{algorithmic}
\end{algorithm}
\begin{figure}
    \centering
    \includegraphics[width=12cm]{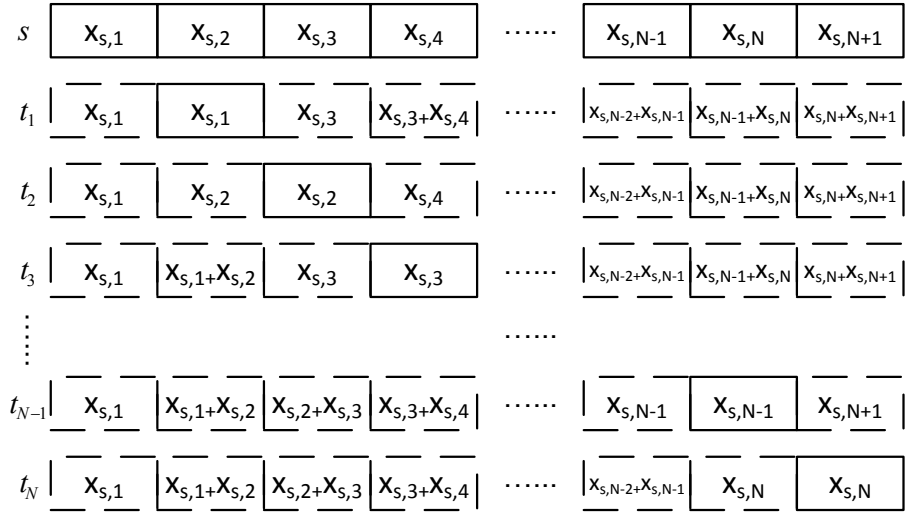}
    \caption{Timing structure after the relay pre-ordering operations for CBC-SSAF strategy, where solid boxes denote transmitted signal and dashed boxes denote received signal.}
    \label{fig:timing}
\end{figure}

From the strategy description, we can easily see that the diversity gains for $t_1$ and $t_N$ are not worse than the diversity gains for other destination nodes. This is because for $t_1$ and $t_N$, there are two slots of source's signal (the first and the last slots, and the first and the second last slots, respectively) not protected by extra paths. While for other destination nodes, three slots of source's signal (the first, the $l$th and the last slots) are not protected by extra paths. For the scenario with only two destination nodes, each one has the same DMT characteristic. Thus, we only focus on the worst achievable DMT for the destination node $t_l$, for $2\leqslant{}l\leqslant{N-1}$, rather than $t_1$ or $t_N$, to give a DMT lower bound.

\begin{theorem}
The achievable DMT lower bound of CBC-SSAF for a CBC with a single source node and $N$ cooperative destination nodes but without any dedicated relay node is
\begin{eqnarray}
d(r)>[(N-3)-(N+1)r]^++(1-\frac{N+1}{N}r)^+.
\label{eqn:dmt}
\end{eqnarray}
\end{theorem}
\begin{proof}
Please refer to Appendix A.
\end{proof}

Compared to currently best known relaying strategy for cooperative broadcast channels CBC-DDF, CBC-SSAF offers tremendous DMT improvement in the high multiplexing gain regime. This is the region where the CBC-DDF strategy is incapable to beat direct transmission strategy. As shown in Fig. \ref{fig:cbcdmt}, the DMT analytical results state that the CBC-SSAF strategy approaches the DMT upper bound as $N$ increases and is thus asymptotically optimal. Moreover, from the outage behaviors in Fig. \ref{fig:outagecbc}, it is easy to see the diversity order of CBC-SSAF is almost dominantly better than those of non-cooperative and CBC-DDF strategies in the high spectral efficiency regime.
\begin{figure}
    \centering
    \subfigure[20 users.]{
    \includegraphics[scale=0.43]{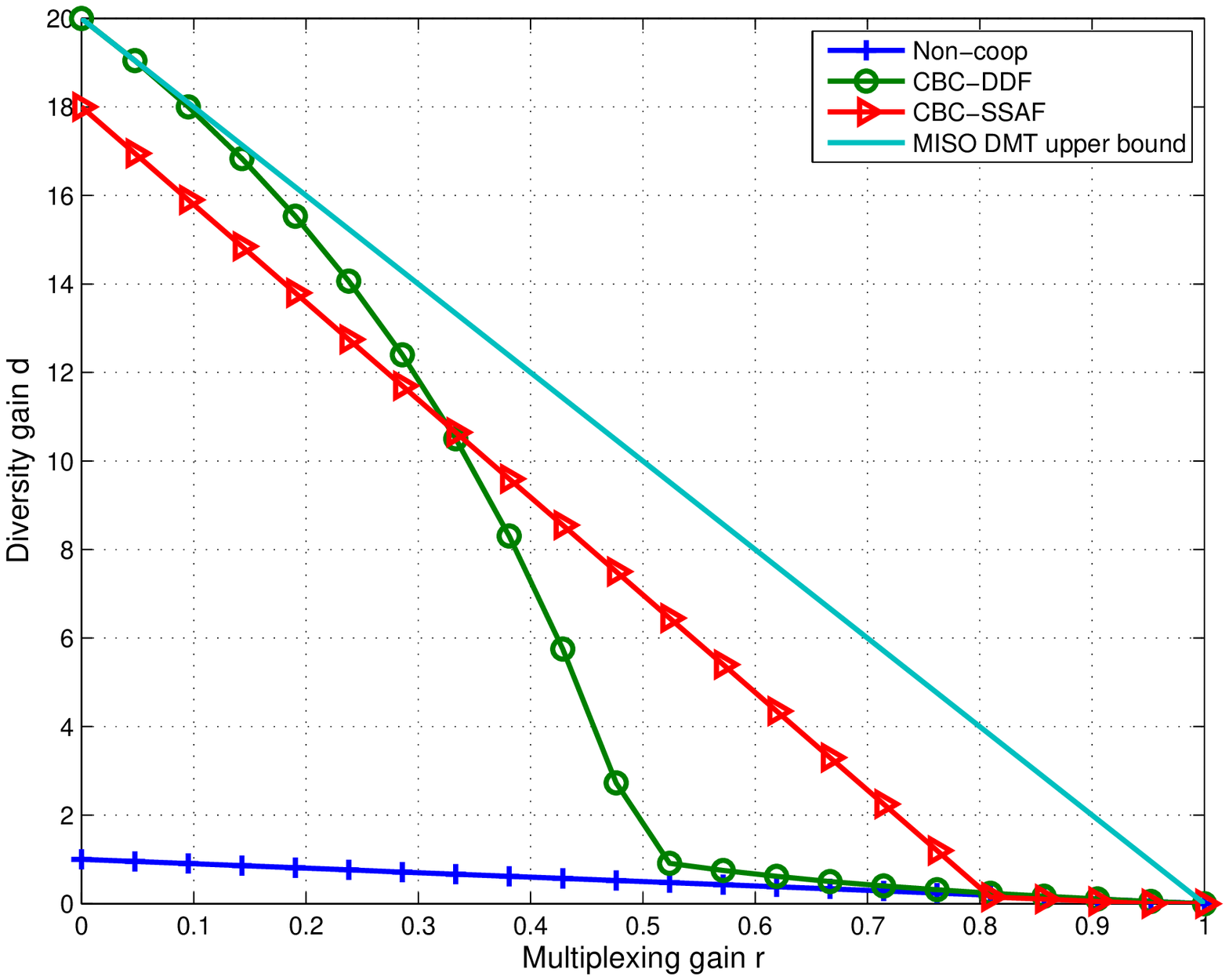}
    }
    \subfigure[50 users.]{
    \includegraphics[scale=0.43]{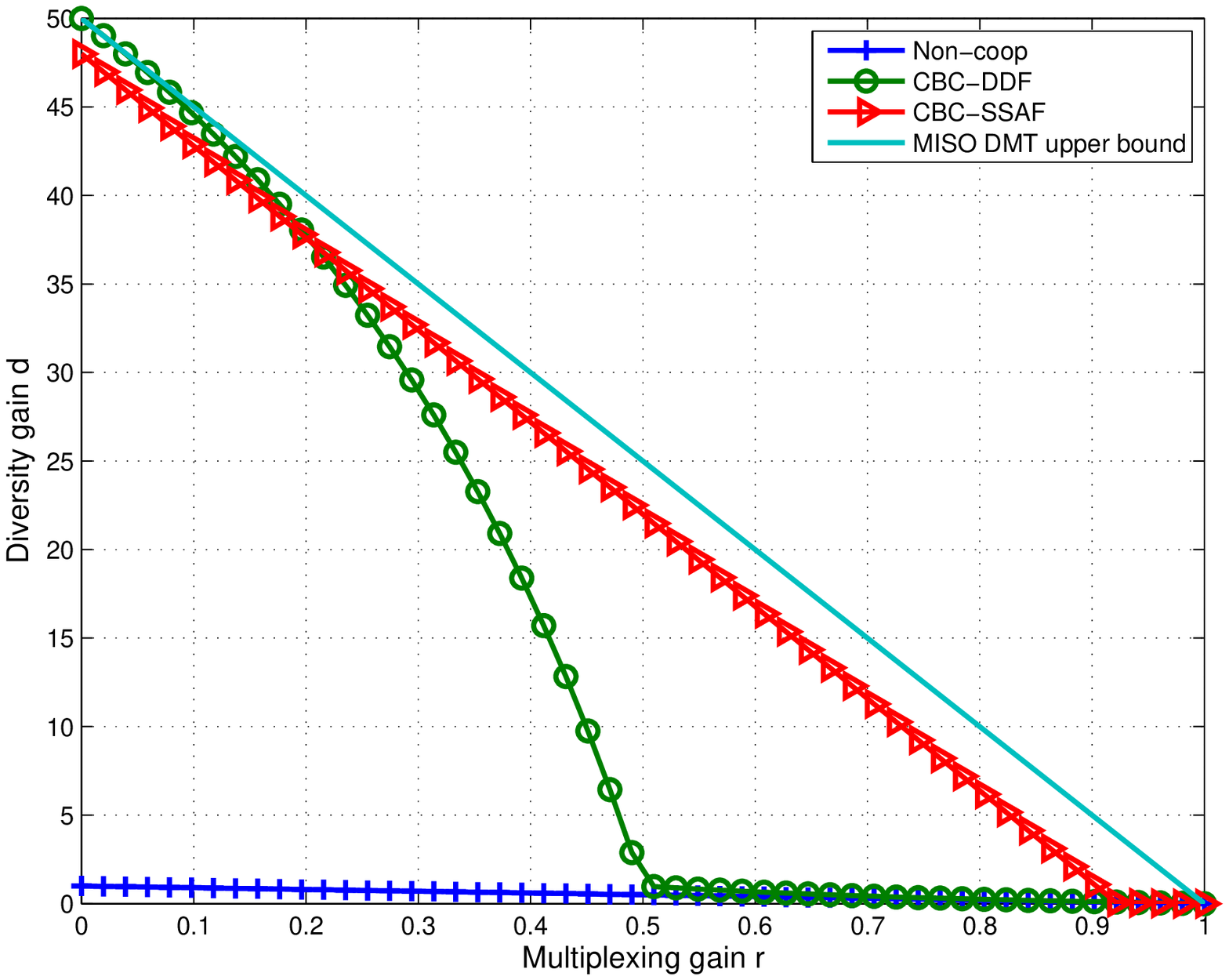}
    }
    \caption{DMTs of non-cooperative, CBC-DDF and CBC-SSAF strategies.}
    \label{fig:cbcdmt}
\end{figure}

It is worth spending some effort to explain the gap between the DMT lower bound and the MISO DMT upper bound. Firstly, we emphasize again \eqref{eqn:dmt} is an achievable DMT lower bound after the relay pre-ordering using Algorithm \ref{algorithm:cbc}, whose purpose is to separate consecutively used relays as much as possible. Secondly, we should note after the relay pre-ordering, the channel gain between two consecutively used relays is chosen as small as possible and we approximate their interaction as small noise enhancement for analytical simplicity which does not affect the DMT. This, however, will result losing some diversity gains compared to the genie-aided MISO case. It is indeed such ordering and approximation that generate a constant gap against the MISO upper bound. The gap is not something inherent for every SSAF strategy (as was shown in \cite{shengyang}), it is however so for CBC-SSAF because of the special structure of CBC and the relay pre-ordering algorithm with approximation. As $N$ becomes large enough, the marginal gain becomes negligible and the ratio between the MISO upper bound and the derived DMT lower bound tends to $1$. Thus, it is safe to say the proposed CBC-SSAF strategy is asymptotically optimal. Without the relay pre-ordering algorithm, channel gains between consecutively used relays might not be small enough to be approximated as being in noise level, and in such cases, calculating the DMT might be mathematically prohibited as stated in \cite{shengyang}. Finally, we should note the relay pre-ordering algorithm is only used to simplify the analysis and the exact DMT of CBC-SSAF lies somewhere between the lower and upper bounds which may not be necessarily dependent on the ordering.
\begin{figure}
    \centering
    \subfigure[3 users.]{
    \includegraphics[scale=0.43]{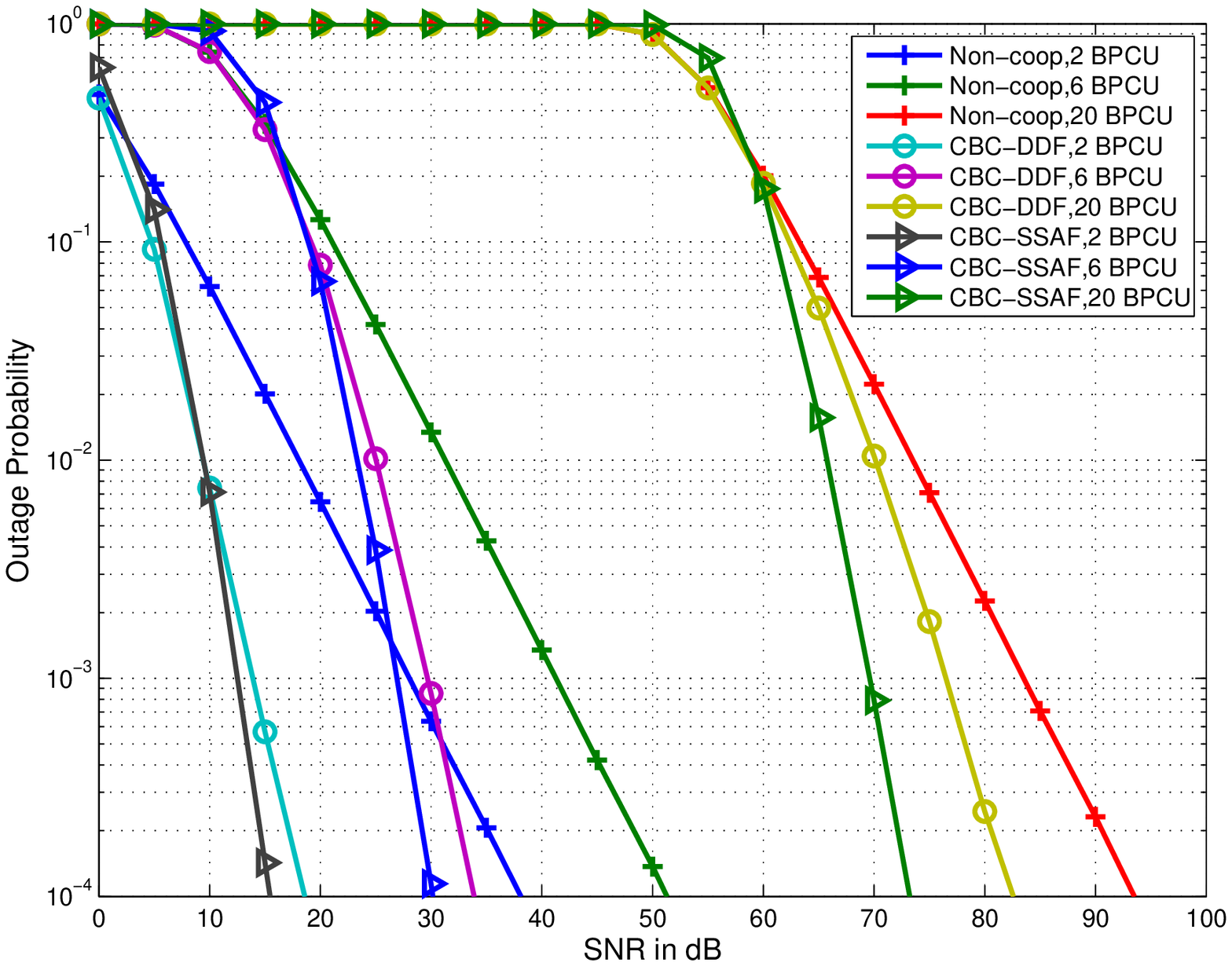}
    }
    \subfigure[10 users.]{
    \includegraphics[scale=0.43]{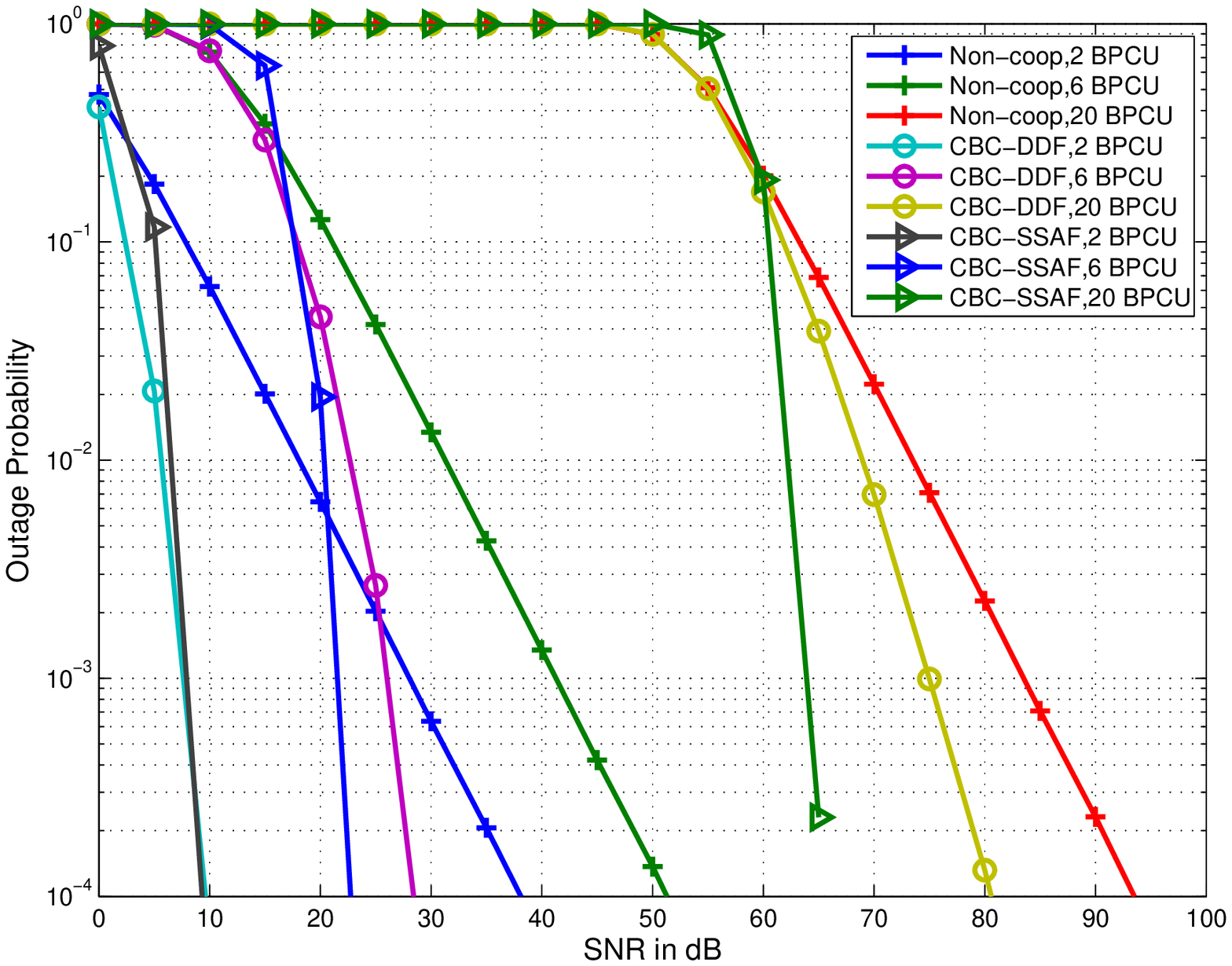}
    }
    \caption{Outage behaviors of non-cooperative, CBC-DDF and CBC-SSAF strategies.}
    \label{fig:outagecbc}
\end{figure}
\section{CMA-SSAF strategy}
\begin{figure}
    \centering
    \includegraphics[width=6cm]{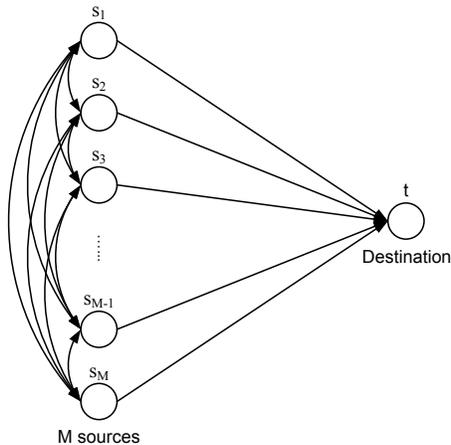}
    \caption{Cooperative multiple access channel with $N$ source nodes and a single destination node.}
    \label{fig:cma}
\end{figure}
Inspired by the fact that the Gaussian multiple access channel and the Gaussian broadcast channel are fundamentally related and are duals to each other \cite{jindal,jafar}, one can naturally conjecture that SSAF strategy, which has been shown to be asymptotically optimal for CBC may also be asymptotically optimal for CMA. In order to prove this conjecture, we slightly modify the original SSAF strategy and propose a CMA-SSAF strategy which will be shown to indeed achieve the DMT upper bound for CMA.

The cooperative multiple access problem under consideration is illustrated in Fig. \ref{fig:cma}. As shown in the figure, $M$ source nodes $(s_1,s_2,......,s_M)$ want to transmit their independent information to a single destination node $t$. Again, we assume the existing physical links are all quasi-static flat Rayleigh-fading. This scenario is similar to a single-cell uplink phase with cooperative users.

Firstly, we point out that in contrast to CBC-SSAF, CMA-SSAF is a full multiplexing gain strategy. In CBC-SSAF strategy, since the source node keeps transmitting a new message in every time slot in order to maximize the multiplexing gain, the destination nodes are busy all the time listening from the source node. If a destination node wants to help others to achieve better performance, it has to sacrifice at least one time slot to transmit but not to listen due to the half duplex constraint. Thus, although CBC-SSAF is asymptotically optimal for CBC, it is not a full multiplexing strategy. In CMA-SSAF strategy, since there is a single destination node $t$, the total achievable multiplexing gain can not be more than $1$. Thus, on average, each user is busy only $1/M$ of the time transmitting its own independent information, and idle $(M-1)/M$ of the time. If we allow each user to help others by relaying its received signal from other source nodes to the destination node in its own idle time, the diversity gain can be improved without losing the achievable multiplexing gain. This explains why CMA-SSAF strategy is optimal with full multiplexing gain.

Secondly, we note that in order to achieve the maximum diversity gain, each source node's messages must be relayed by every other source node. Thus, each source node needs to help others by relaying its received signal from other source nodes ( $(M-1)$ in total) to the destination node. Fortunately, as each source node is idle $(M-1)/M$ of the time, this task can be done without losing its achievable multiplexing gain. In order to reduce the scheduling complexity, we allow a source node to act as a relay node and forward every other source node's signal it received to the destination node, in a single time slot. The main purpose is to maximize the use of nonorthogonal signal subspace in order to save the effort for complex scheduling, and because there is a single destination node, the use nonorthogonal signal subspace does not reduce the achievable multiplexing gain.

Now, we detail the proposed CMA-SSAF strategy as follows:
\begin{enumerate}
\item One cooperation frame contains $2M$ time slots.
\item In the $l$th time slot, for $l=1,2,...,M$, source node $s_l$ transmits a new message $x_{s_l,l}$, other source nodes receive noisy versions of the same message in the same time slot.
\item In the $(M+l)$th time slot, for $l=1,2,...,M$, select source node $s_l$ as the $l$th relay node and let it transmit signal which is combination of its own new message $x_{s_l,M+l}$ and its received signal from other source nodes in the first half of the cooperation frame.
\end{enumerate}

Of course, we need to divide the total power to a source node's own signal and its relayed signal. However, the performance difference by using different power allocation schemes can only be told in the low SNR regime. In the high SNR regime, because the number of users is finite, a simple equal power allocation scheme suffices to generate the same DMT as that of an optimal power allocation scheme. Thus, in a source node's relaying time slot, we use an equal power allocation scheme which allocates equal power to a source node's own signal and every other signal it relays.
\begin{theorem}
The DMT of CMA-SSAF strategy for a CMA with $M$ cooperative source nodes and a single destination node but without any dedicated relay node is
\begin{eqnarray}
d(r)=M(1-r).
\end{eqnarray}
\end{theorem}
\begin{proof}
Please refer to Appendix B.
\end{proof}
\begin{figure}
    \centering
    \subfigure[5 users.]{
    \includegraphics[scale=0.43]{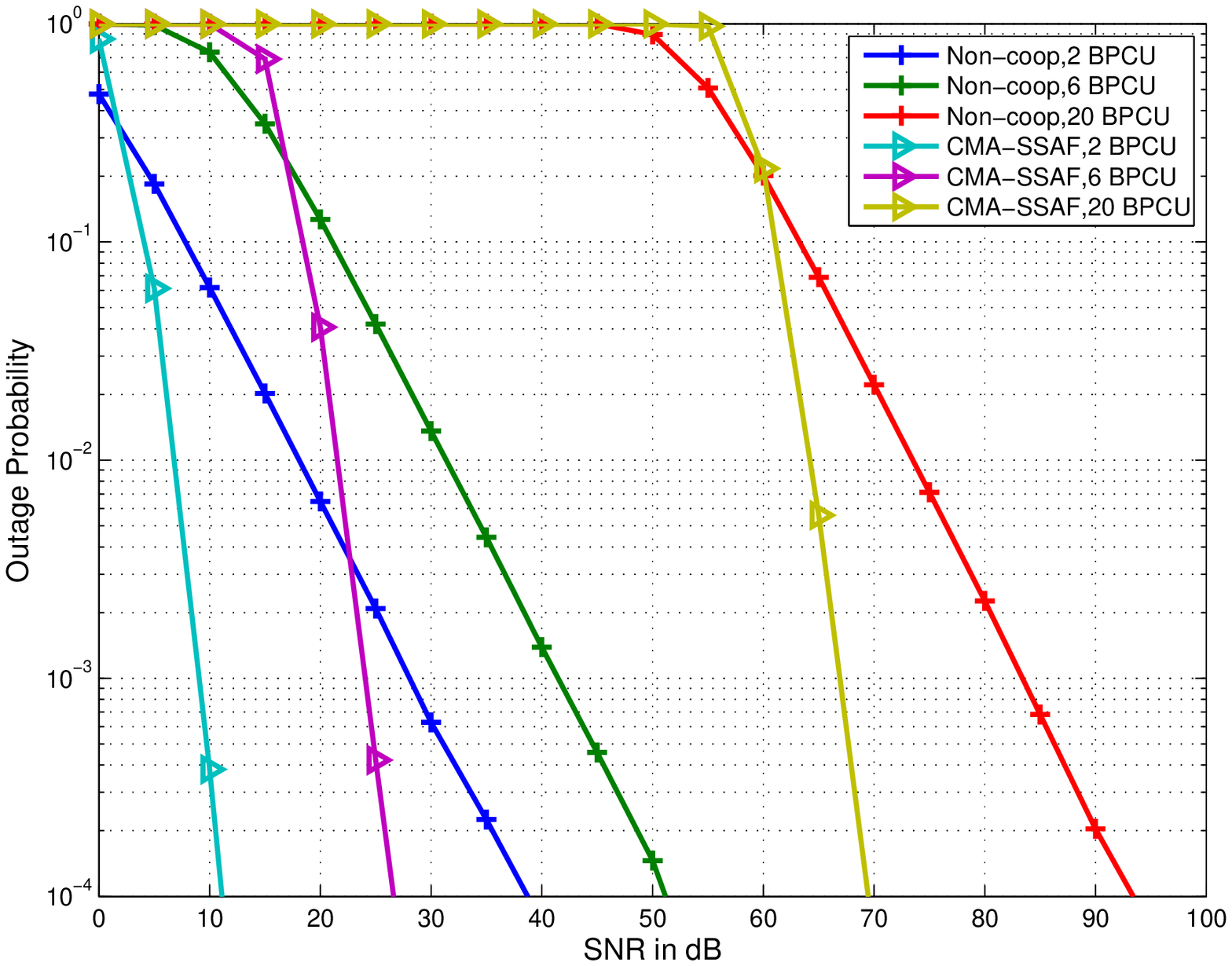}
    }
    \subfigure[10 users.]{
    \includegraphics[scale=0.43]{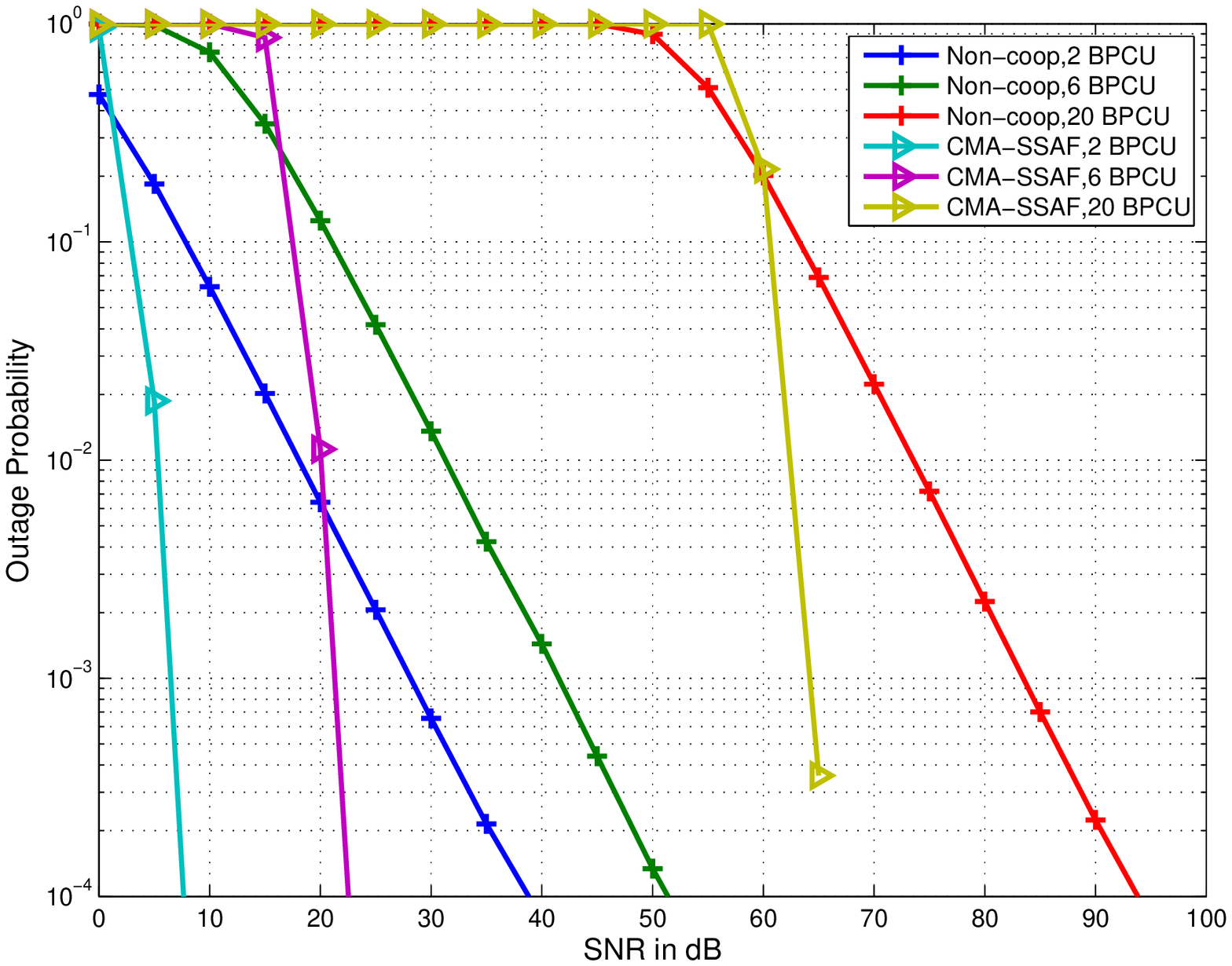}
    }
    \caption{Outage probabilities of non-cooperative and CMA-SSAF strategies.}
    \label{fig:outagecma}
\end{figure}
The CMA-SSAF strategy can exactly achieve the optimal DMT upper bound with any number of source nodes and finite cooperation frame length, and thus more suitable for real-time applications with any number of cooperative users. Fig. \ref{fig:outagecma} shows the outage behaviors of the CMA-SSAF strategy for CMAs with various numbers of cooperative users. It can be easily seen that the diversity order of CMA-SSAF is much higher than that of the non-cooperative strategy. Moreover, the slops of the outage curves remains the same even in the high spectral efficiency regime, which demonstrates the optimality of the CMA-SSAF strategy.
\section{Conclusion}
In this paper, we studied the class of SSAF strategies, which was previously used to achieve the DMT upper bound for CMR. We extended the use of SSAF strategy to CBC and CMA, and proposed two SSAF variants, i.e., CBC-SSAF and CMA-SSAF strategies for these two channel models respectively. Together with the fact that SSAF strategy is also optimal for CMR, we have proved that SSAF strategy can be used as a strategy which is universally optimal for all these three classical channel models. Previous research showed that the best known strategies for CMR, CBC and CMA are SSAF, DDF and NAF respectively. In practice, a wireless node can be a source node, a relay node or a destination node from time to time. Thus, from the whole network perspective, it is clearly not convenient in implementation to use a strategy which can only provide optimality for a particular channel model because a wireless node may have different roles and functions at different time. Besides, it is not attractive to frequently switch between the class of amplify and forward strategies and class of decode and forward strategies. Thus, finding a transmission strategy which is universally optimal for all these three channel models is very important. The main idea of SSAF strategy is to assign a node one or more time slots to transmit its own information and one or more time slots to amplify and forward its previously received signal, sequentially and with maximum use of the nonorthogonal signal subspace. Using such rules, one can design a SSAF variant which is optimal for every of those three classical channel models, i.e., CMR, CBC and CMA, which shows SSAF strategy has great potential to provide universal optimality for wireless cooperative networks.

For more general network topologies, e.g., multiple-source multiple-destination cooperative networks, how the SSAF strategy performs is a challenging yet interesting problem. Take a cooperative interference network for instance. The achievable optimal DMT for such a network itself is still an open problem. Intuitively, while naively applying SSAF at both transmitter and receiver sides can provide a destination node near-maximum diversity gain, its achievable multiplexing gain might not be optimal due to the non-desired interference problem. Thus, we conjecture SSAF might need to be used together with other techniques like interference alignment to achieve optimal DMT. This is no doubt a very interesting direction for future research.
\appendices
\section{Proof of Theorem 1}
Since there is no difference in processing different symbols, we assume each transmission contains only one symbol. The received signal vector for the destination node $t_l$ is
\begin{eqnarray}
\textbf{y}_{t_l}=\textbf{H}_{t_l,(N+1)\times{(N+1)}}\cdot{}\textbf{x}_{t_l,(N+1)\times1}+\textbf{n}_{t_l,(N+1)\times1}.
\end{eqnarray}

Assume the channel gain between $s$ and $t_l$ is $h_{s,t_l}$, for $1\leqslant{l}\leqslant{N}$, and the channel gain between $t_l$ and $t_{l+1}$ is $h_{t_l,t_{l+1}}$, for $1\leqslant{l\leqslant{N-1}}$. The relay pre-ordering algorithm aims to choose $|h_{t_l,t_{l+1}}|$ as small as possible. Under independent Rayleigh-fading channel realizations, when the number of destination nodes is large,  there is a high probability that a bad CER to NER link exists with $|h_{t_l,t_{l+1}}|$ being very small, for $1\leqslant{l\leqslant{N-1}}$. Thus, the signal from CER is small interference at NER, which can be viewed as small noise enhancement at NER. From the analysis for the accumulated noise later \eqref{eqn:noisedmt}, we know this small noise enhancement does not affect the DMT analysis. Thus, for analytical simplicity, we assume $|h_{t_l,t_{l+1}}|=0$. Under reciprocal channel realizations, we also have $|h_{t_{l+1},t_{l}}|=|h_{t_l,t_{l+1}}|=0$. Finally, due to the half-duplex constraint, we know that $|h_{t_l,t_l}|=0$, for $1\leqslant{l\leqslant{N}}$.

From the CBC-SSAF strategy description earlier, we know that $y_{t_l,l+1}=0$. Thus, the effective received signal vector changes to
\begin{eqnarray}
\tilde{\textbf{y}}_{t_l}=\tilde{\textbf{H}}_{t_l,N\times{N}}\cdot{}\tilde{\textbf{x}}_{t_l,N\times{1}}+\tilde{\textbf{n}}_{t_l,N\times{1}}
\label{eqn:effectivechannel}
\end{eqnarray}
where
\begin{eqnarray}
\tilde{\textbf{H}}_{{t_l},N\times{N}}=
\begin{pmatrix}
\textbf{A}_{{t_l},l\times{l}} & \textbf{0}_{{t_l},l\times{(N-l)}} \\
\textbf{0}_{{t_l},(N-l)\times{l}} & \textbf{B}_{{t_l},(N-l)\times{(N-l)}} \\
\end{pmatrix}_{N\times{N}},
\end{eqnarray}
\begin{eqnarray}
\textbf{A}_{t_l}=
\begin{pmatrix}
h_{s,t_l} & 0 & 0 & ... & 0 \\
h_{t_1,t_l}\beta_{t_1}h_{s,t_1} & h_{s,t_l} & 0 & ... & 0 \\
\multicolumn{5}{c}{......} \\
0 & ... & h_{t_{l-2},t_l}\beta_{t_{l-2}}h_{s,t_{l-2}} & h_{s,t_l} & 0\\
0 & ... & 0 & 0 & h_{s,t_l}
\end{pmatrix},
\end{eqnarray}
and
\begin{eqnarray}
\textbf{B}_{t_l}=
\begin{pmatrix}
h_{s,t_l} & 0 & ... & 0 \\
h_{t_{l+2},t_l}\beta_{t_{l+2}}h_{s,t_{l+2}} & h_{s,t_l} & ... & 0 \\
\multicolumn{4}{c}{......} \\
0 & ... & h_{t_N,t_l}\beta_{t_N}h_{s,t_N} & h_{s,t_l}
\end{pmatrix}.
\end{eqnarray}
$\beta_{t_l}$ denotes the normalization factor at destination node $t_l$ (the selected $l$th relay node) so that its forwarded signal satisfies its average energy constraint $E$.

Thus, we have:
\begin{enumerate}
\item For time slots $k=1,l$ and $(l+2)$
\begin{eqnarray}
y_{t_l,k}=h_{s,t_l}\cdot{}x_{s,k}+n_{t_l,k}.
\label{eqn:channel1}
\end{eqnarray}
\item For time slots $2\leqslant{}k\leqslant{}(N+1)$, $k\neq{}l,l+1$ and $(l+2)$
\begin{eqnarray}
y_{t_l,k} &= & \check{\textbf{H}}_{t_l,k}\cdot{}\check{\textbf{x}}_{t_l,k}+n_{t_l,k} \notag\\
&=&
\begin{pmatrix}
h_{t_{k-1},t_l}\beta_{t_{k-1}}h_{s,t_{k-1}} & h_{s,t_l}
\end{pmatrix}
\begin{pmatrix}
x_{s,k-1} \\
x_{s,k}
\end{pmatrix}+n_{t_l,k}.
\label{eqn:channel2}
\end{eqnarray}
\end{enumerate}

Note that, in the second case above, the term $n_{t_l,k}$ is actually the accumulated noise from both relay and destination nodes, which can be represented as
\begin{eqnarray}
n_{t_l,k}=\hat{n}_{t_l,k}+h_{t_{k-1},t_l}\beta_{t_{k-1}}n_{t_{k-1},k-1}
\label{eqn:anoise}
\end{eqnarray}
where the normalization factor $\beta_{t_{k-1}}$ at the relay node $t_{k-1}$ should satisfy the energy constraint
\begin{eqnarray}
|\beta_{t_{k-1}}|^2\leqslant\frac{E}{E|h_{s,t_{k-1}}|^2+\sigma^2}=\frac{\rho}{\rho|h_{s,t_{k-1}}|^2+1}.
\label{eqn:const}
\end{eqnarray}

Let $h$ be complex standard normal distributed and $v$ denotes the exponential order of $\frac{1}{|h|^2}$. The probability density function (pdf) of $v$ can be written as \cite{azarian}
\begin{eqnarray}
p_v\dot{=}
  \left\{
    \begin{array}{l l}
      \rho^{-\infty}=0, & \textrm{for} \:v<0\\
      \rho^{-v}, & \textrm{for} \:v\geqslant{0}.\\
    \end{array}
    \right.
\end{eqnarray} Thus, for $N$ independent identically distributed (i.i.d) variables $\{v_j\}_{j=1}^N$, the probability that $(v_1,...,v_N)$ belongs to a set $O$ is
\begin{eqnarray}
P_O\dot{=}\rho^{-d_O}, \quad \textrm{ for } \: d_O=\inf_{(v_1,...,v_N)\in{O^+}}\sum_{j=1}^Nv_j
\label{eqn:do}
\end{eqnarray} given that $\mathbb{R}^{N+}$ denotes the set of nonnegative $N$-tuples and $O^+=O\bigcap{}\mathbb{R}^{N+}$ is not empty. So, the exponential order of $P_O$ depends only on $O^+$ and is dominated by the realization with the largest exponential order.

From \cite{my1}, it can be shown that under the consideration of outage events belonging to set $O^+$, proper selection of the normalization factor $\beta_{t_{k-1}}$ will make its exponential order vanish in all the DMT analytical expressions and the noise enhancement problem does not affect the DMT result. Thus, the DMT of our proposed strategy depends only on the channel matrix and not on the variance of the accumulated noise. So, for analytical simplicity, we assume the accumulated noise equals to the noise at each destination node which does not affect the DMT analysis.

In order to get a DMT lower bound, we want to first upper-bound the probability of error of the ML decoder. Using Bayes' theorem, we can write
\begin{eqnarray}
P_E(\rho) & = & P_O(R)P_{E|O}+P_{E,O^c} \notag\\
& \leqslant & P_O(R)+P_{E,O^c}
\end{eqnarray}
where the outage events set $O$ and its complement set $O^c$ are chosen such that $P_O(R)$ dominates $P_{E,O^c}$, i.e.,
\begin{eqnarray}
P_{E,O^c}\dot{\leqslant{}}P_O(R).
\label{eqn:condition}
\end{eqnarray}
Thus, we have
\begin{eqnarray}
P_E(\rho)\dot{\leqslant{}}P_O(R).
\label{eqn:condd}
\end{eqnarray}
From \eqref{eqn:do} and the simplified channels described by \eqref{eqn:channel1} and \eqref{eqn:channel2}, we let
\begin{eqnarray}
P_O(R)\dot{=}\rho^{-d_O(r)},
\label{eqn:po}
\end{eqnarray}
for
\begin{eqnarray}
d_O(r)=\displaystyle{\inf_{(\textbf{v},\textbf{u})\in{O^{+}}}}[\displaystyle{\sum_{k=2,k\neq{l,l+1,l+2}}^{N+1}(v_{k-1}+u_{k-1})}+v_l].
\label{eqn:dor}
\end{eqnarray}
Thus, $d_O(r)$ provides a lower bound on the diversity gain achieved by our proposed CBC-SSAF strategy for destination node $t_l$.

From the definition of the outage probability \cite{telatar}, for the $N+1$ messages from the source node $s$ to the destination node $t_l$, we know that
\begin{eqnarray}
P_O(R)&=&P[I(\tilde{\textbf{x}}_{t_l};\tilde{\textbf{y}_{t_l}})<R] \notag\\
&\leqslant&P[\sum_{k=2,k\neq{l,l+1,l+2}}^{N+1}I(\check{\textbf{x}}_{t_l,k};y_{{t_l},k})+\sum_{k=1,l,l+2}I(x_{s,k};y_{{t_l},k})<(N+1)r\log\rho].
\label{eqn:outage}
\end{eqnarray}
For time slots $k=1,l$ and $(l+2)$
\begin{eqnarray}
\lim_{\rho\rightarrow\infty}\frac{I(x_{s,k};y_{t_l,k})}{\log\rho}=\lim_{\rho\rightarrow\infty}\frac{\log(1+\rho|h_{s,t_l}|^2)}{\log\rho}=(1-v_l)^{+}.
\label{eqn:11}
\end{eqnarray}
For time slots $2\leqslant{}k\leqslant{}N+1,k\neq{l,l+1,l+2}$
\begin{eqnarray}
&\,&\lim_{\rho\rightarrow\infty}\frac{I(\check{\textbf{x}}_{t_l,k};y_{t_l,k})}{\log\rho}=\lim_{\rho\rightarrow\infty}\frac{\log(1+\rho\check{\textbf{H}}_{t_l,k}\check{\textbf{H}}_{t_l,k}^{\dag})}{\log\rho}\notag\\
&=&\lim_{\rho\rightarrow\infty}\frac{\log(1+\rho|h_{t_{k-1},t_l}\beta_{t_{k-1}}h_{s,t_{k-1}}|^2+\rho|h_{s,t_l}|^2)}{\log\rho} \notag\\
&=&(\max\{1-v_{k-1}-u_{k-1},1-v_l\})^+
\label{eqn:22}
\end{eqnarray}
where $[\cdot]^\dag$ is used to denote the matrix conjugated transposition operation. Thus, from \eqref{eqn:outage}, \eqref{eqn:11} and \eqref{eqn:22}, the outage events set $O^+$ should be defined as
\begin{eqnarray}
&\,&\!\!\!\!\!\!\!\!\!O^+=\{(\textbf{v},\textbf{u})\in\mathbb{R}^{(2N-5)+}|3(1-v_l)^+\notag\\
&\,&\quad\:+\sum_{k=2,k\neq{l,l+1,l+2}}^{N+1}(\max\{1-v_{k-1}-u_{k-1},1-v_l\})^+<(N+1)r\}.
\label{eqn:condition11}
\end{eqnarray}

From \eqref{eqn:condition11}, we can easily see that, in order to let the outage events happen, the following constraints must be simultaneously satisfied (necessary but not sufficient conditions):
\begin{enumerate}
\item For $v_l$, we have
\begin{eqnarray}
v_l>(1-\frac{N+1}{N}r)^+.
\label{eqn:condition1}
\end{eqnarray}
\item For $\sum_{k=2,k\neq{l,l+1,l+2}}^{N+1}(v_{k-1}+u_{k-1})$, we have
\begin{eqnarray}
\sum_{k=2,k\neq{l,l+1,l+2}}^{N+1}(v_{k-1}+u_{k-1})>[(N-3)-(N+1)r]^+.
\label{eqn:condition2}
\end{eqnarray}
\end{enumerate}

From \eqref{eqn:condition} and \eqref{eqn:condd}, we see that \eqref{eqn:condition1} and \eqref{eqn:condition2} can be used to lower-bound the diversity gain $d(r)$ if and only if the outage events described by \eqref{eqn:condition11} dominate the probability of error of the ML decoder. To see this, we first observe that the channel described by \eqref{eqn:effectivechannel} can be seen as a coherent linear Gaussian channel as
\begin{eqnarray}
\tilde{\textbf{y}}_{t_l,N\times1}=
\tilde{\textbf{H}}_{t_l,N\times{N}}\cdot{}
\tilde{\textbf{x}}_{t_l,N\times1}
+\tilde{\textbf{n}}_{t_l,N\times1}.
\end{eqnarray}
Thus, the average pairwise error probability (PEP) of the ML decoder at high SNR regime can be upper-bounded by $[3,\text{ Eqn. } 7]$
\begin{eqnarray}
P_{PE}&\leqslant{}&\det(\mathbf{I}_N+\frac{1}{2}\rho\tilde{\textbf{H}}_{t_l}\tilde{\textbf{H}}_{t_l}^\dagger)^{-1} \notag \\
&\dot{=}&\det(\mathbf{I}_N+\rho\tilde{\textbf{H}}_{t_l}\tilde{\textbf{H}}_{t_l}^\dagger)^{-1}
\end{eqnarray}
where we set $\sum_{\tilde{\textbf{n}}_{t_l}}=\sigma^2\mathbf{I}_N$ because such manipulation does not affect the DMT. From \eqref{eqn:11}, \eqref{eqn:22} and the formula for multiple-input multiple-output channel capacity \cite{telatar,foschini}, we know that
\begin{eqnarray}
&\,&[\sum_{k=2,k\neq{l,l+1,l+2}}^{N+1}(\max\{1-v_{k-1}-u_{k-1},1-v_l\})^++3(1-v_l)^+]\log\rho \notag\\
&\dot{=}&\sum_{k=1,l,l+2}I(x_{s,k};y_{t_l,k})+\sum_{k=2,k\neq{l,l+1,l+2}}^{N+1}I(\check{\textbf{x}}_{t_l,k};y_{t_l,k}) \notag \\ &\leqslant{}&\log[\det(\mathbf{I}_N+\rho\tilde{\textbf{H}}_{t_l}\tilde{\textbf{H}}_{t_l}^\dagger)].
\end{eqnarray}
Thus,
\begin{eqnarray}
\rho^{\sum_{k=2,k\neq{l,l+1,l+2}}^{N+1}(\max\{1-v_{k-1}-u_{k-1},1-v_l\})^++3(1-v_l)^+}\dot{\leqslant}\det(\mathbf{I}_N+\rho\tilde{\textbf{H}}_{t_l}\tilde{\textbf{H}}_{t_l}^\dagger)
\end{eqnarray}
and we can upper-bound $P_{PE}$ as
\begin{eqnarray}
P_{PE}&\dot{\leqslant}&\det(\mathbf{I}_N+\rho\tilde{\textbf{H}}_{t_l}\tilde{\textbf{H}}_{t_l}^\dagger)^{-1}\notag\\
&\dot{\leqslant}&\rho^{-\sum_{k=2,k\neq{l,l+1,l+2}}^{N+1}(\max\{1-v_{k-1}-u_{k-1},1-v_l\})^+-3(1-v_l)^+}.
\end{eqnarray}

Let the $N+1$ messages from the source node $S$ to the destination node $r_i$ form a codeword of length $N+1$. The data rate is $R=(N+1)r\log\rho$ BPCU and we have a total of $\rho^{(N+1)r}$ codewords. Thus, we can bound $P_E(\rho)$ as
\begin{eqnarray}
P_E(\rho)&\dot{\leqslant}&\rho^{-\sum_{k=2,k\neq{l,l+1,l+2}}^{N+1}(\max\{1-v_{k-1}-u_{k-1},1-v_l\})^+-3(1-v_l)^++(N+1)r}.
\end{eqnarray}
Therefore, $P_{E,O^c}$ can be written as \cite{azarian}
\begin{eqnarray}
P_{E,O^c}\dot{\leqslant}\int_{O^{c+}}\rho^{-d_e(r,\textbf{v},\textbf{u})}\textbf{d}_{\textbf{v}}\textbf{d}_\textbf{u}
\end{eqnarray}
where
\begin{eqnarray}
&\,&\!\!\!\!\!\!\!\!\!\!\!\!\!\!d_e(r,\textbf{v},\textbf{u})=-(N+1)r+3(1-v_l)^+\notag\\
&\,&\,\,\,\,\,\,\,+\sum_{k=2,k\neq{l,l+1,l+2}}^{N+1}(\max\{1-v_{k-1}-u_{k-1},1-v_l\})^+ \notag \\
&\,&\,\,\,\,\,\,\,+[\displaystyle{\sum_{k=2,k\neq{l,l+1,l+2}}^{N+1}(v_{k-1}+u_{k-1})}+v_l].
\end{eqnarray}
Because $P_{E,O^c}$ is dominated by the smallest term of $d_e(r,\textbf{v},\textbf{u})$ over $O^{c+}$, we can write
\begin{eqnarray}
P_{E,O^c}\dot{\leqslant}\rho^{-d_e(r)}, \quad \text{ for } d_e(r)=\inf_{(\textbf{v},\textbf{u})\in{O^{c+}}}d_e(r,\textbf{v},\textbf{u}).
\label{eqn:der}
\end{eqnarray}
Comparing \eqref{eqn:po} and \eqref{eqn:der}, we see that for \eqref{eqn:condition} to be met, $O^+$ should be defined as
\begin{eqnarray}
&\,&\!\!\!\!\!\!\!\!\!O^+=\{(\textbf{v},\textbf{u})\in\mathbb{R}^{(2N-5)+}|3(1-v_l)^+\notag\\
&\,&\quad\:+\sum_{k=2,k\neq{l,l+1,l+2}}^{N+1}(\max\{1-v_{k-1}-u_{k-1},1-v_l\})^+\leqslant{}(N+1)r\}
\end{eqnarray}
which contains the outage events set \eqref{eqn:condition11}. So, we conclude that the outage events described by \eqref{eqn:condition11} also satisfy \eqref{eqn:condition} and therefore dominate the probability of error of the ML decoder. Thus, we can use \eqref{eqn:condition1} and \eqref{eqn:condition2} to lower-bound $d_O(r)$ (which further provides a lower bound for $d(r)$) as
\begin{eqnarray}
d(r)\geqslant{}d_O(r)>[(N-3)-(N+1)r]^++(1-\frac{N+1}{N}r)^+.
\end{eqnarray}
\section{Proof of Theorem 2}
Similar to the reasons for CBC, we assume each transmission contains only one symbol. The received signal vector at the destination node $t$ can be written as
\begin{eqnarray}
\textbf{y}_t=\textbf{H}_{t,2M\times{2M}}\cdot\textbf{x}_{2M\times{1}}+\textbf{n}_{t,2M\times1}\notag\\
\label{eqn:cmasignal}
\end{eqnarray}
where
\begin{eqnarray}
\textbf{H}_t=
\begin{pmatrix}
\textbf{A} & \textbf{0} \\
\textbf{B} & \frac{1}{\sqrt{M}}\textbf{A}
\end{pmatrix},
\end{eqnarray}
\begin{eqnarray}
\textbf{A}=
\begin{pmatrix}
h_{s_1,t} & 0 & 0 & ... & 0 & 0 \\
0 & h_{s_2,t} & 0 & ... & 0 & 0 \\
0 & 0 & h_{s_3,t} & ... & 0 & 0 \\
\multicolumn{6}{c}{......} \\
0 & 0 & 0 & 0 & ... & h_{s_M,t}
\end{pmatrix}
\end{eqnarray}
and
\begin{eqnarray}
\textbf{B}=\frac{1}{\sqrt{M}}
\begin{pmatrix}
0 & h_{s_1,t}\beta_{1,2}h_{s_2,s_1} & h_{s_1,t}\beta_{1,3}h_{s_3,s_1} & ... & h_{s_1,t}\beta_{1,M}h_{s_M,s_1}\\
h_{s_2,t}\beta_{2,1}h_{s_1,s_2} & 0 & h_{s_2,t}\beta_{2,3}h_{s_3,s_2} & ... & h_{s_2,t}\beta_{2,M}h_{s_M,s_2}\\
\multicolumn{5}{c}{......} \\
h_{s_M,t}\beta_{M,1}h_{s_1,s_M} & h_{s_M,t}\beta_{M,2}h_{s_2,s_M} & ... & h_{s_M,t}\beta_{M,M-1}h_{s_{M-1},s_M} & 0\\
\end{pmatrix}.
\end{eqnarray}
$\beta_{l,k}$, for $1\leqslant{l}\neq{k}\leqslant{M}$, is the normalization factor for the signal from source node $s_k$ to source node $s_l$, so that the scaled signal satisfies the average energy constraint $E$. The term $\frac{1}{\sqrt{M}}$ comes from the equal power allocation scheme. Thus, we have:
\begin{enumerate}
\item In the $l$th time slot, for $1\leqslant{l}\leqslant{M}$,
\begin{eqnarray}
y_{t,l}=h_{s_l,t}x_{s_l,l}+n_{t,l}.
\label{eqn:lth}
\end{eqnarray}
\item In the $(M+l)$th time slot, for $1\leqslant{l}\leqslant{M}$,
\begin{eqnarray}
y_{t,M+l}&=&\check{\textbf{H}}_{M+l}\cdot\check{\textbf{x}}_{M+l}+n_{t,M+l}\notag\\
&=&\frac{1}{\sqrt{M}}h_{s_l,t}(x_{s_l,M+l}+\sum_{k=1,k\neq{l}}^{M}\beta_{l,k}h_{s_k,s_l}x_{s_k,k})+n_{t,M+l},
\label{eqn:m+lth}
\end{eqnarray}
\end{enumerate}
where
\begin{eqnarray}
\check{\textbf{H}}_{M+l}=\frac{1}{\sqrt{M}}h_{s_l,t}
\begin{pmatrix}
1 & \beta_{l,1}h_{s_1,s_l} & \beta_{l,2}h_{s_2,s_l} & ... & \beta_{l,l-1}h_{s_{l-1},s_l} & \beta_{l,l+1}h_{s_{l+1},s_l} & ... & \beta_{l,M}h_{s_M,s_l}
\end{pmatrix},
\end{eqnarray}
\begin{eqnarray}
\check{\textbf{x}}_{M+l}=
\begin{pmatrix}
x_{s_l,M+l} & x_{s_1,1} & x_{s_2,2} & ... & x_{s_{l-1},l-1} & x_{s_{l+1},l+1} & ... & x_{s_M,M}
\end{pmatrix}^T
\end{eqnarray}
and $[\cdot]^T$ is used to denote the matrix transposition operation.

Note that the noise terms in \eqref{eqn:cmasignal} and \eqref{eqn:m+lth} are actually the accumulated noise from both the source nodes themselves and the relayed signal from other source nodes. Similar to the analysis for the accumulated noise in CBC-SSAF, the exponential order of $\beta_{l,k}$, for $1\leqslant{l}\neq{k}\leqslant{M}$, will vanish in all the DMT analytical expressions if we choose the normalization factors properly. Thus, the DMT of our proposed CMA-SSAF strategy depends only on the channel matrix and not on the variance of the accumulated noise.

It is easy to see that the error probability of the joint ML decoder upper-bounds the error probabilities of the source-specific ML decoders and provides a lower bound on each source-specific diversity gain \cite{azarian}. Thus, the error probability of the joint ML decoder also serves a lower bound on the overall achievable diversity gain and we only consider using this worst-performance decoder to give a DMT lower bound.
From \eqref{eqn:do} and the simplified channels described by \eqref{eqn:lth} and \eqref{eqn:m+lth}, we can write:
\begin{eqnarray}
P_O(R)\dot=\rho^{-d_O(r)},
\end{eqnarray}
for
\begin{eqnarray}
d_O(r)=\inf_{\textbf{v,u}\in{O^+}}\sum_{l=1}^M[v_l+\sum_{k=1,k\neq{l}}^Mu_{k,l}],
\label{eqn:outageexpression}
\end{eqnarray}
where $v_l$ denotes the exponential order of $\frac{1}{|h_{s_l,t}|^2}$ and $u_{k,l}$ denotes the exponential order of $\frac{1}{|h_{s_k,s_l}|^2}$. Thus, $d_O(r)$ provides a lower bound on the diversity gain achieved by CMA-SSAF strategy for any of the source nodes.

Because we want to use the joint ML decoder to give a DMT lower bound, for the $2M$ messages transmitted in the $2M$ time slots in one cooperation frame, the outage probability of the joint ML decoder can be written as
\begin{eqnarray}
P_O(R)&=&P[I(\textbf{x};\textbf{y}_{t})<R]\notag\\
&\leqslant{}&P[\sum_{l=1}^MI(x_{s_l,l};y_{t,l})+\sum_{l=1}^MI(\check{\textbf{x}}_{M+l};y_{t,M+l})<2Mr\log\rho].
\end{eqnarray}
From \eqref{eqn:lth}, we know that in the $l$th time slot, for $1\leqslant{l}\leqslant{M}$,
\begin{eqnarray}
\lim_{\rho\rightarrow\infty}\frac{I(x_{s_l,l};y_{t,l})}{\log\rho}=\lim_{\rho\rightarrow\infty}\frac{\log(1+\rho|h_{s_l,t}|^2)}{\log\rho}=(1-v_l)^+.
\label{eqn:lthd}
\end{eqnarray}
From \eqref{eqn:m+lth}, we know that in the $M+l$th time slot, for $1\leqslant{l}\leqslant{M}$,
\begin{eqnarray}
&&\lim_{\rho\rightarrow\infty}\frac{I(\check{\textbf{x}}_{M+l};y_{t,M+l})}{\log\rho}=\lim_{\rho\rightarrow\infty}\frac{\log(1+\rho\check{\textbf{H}}_{M+l}\check{\textbf{H}}_{M+l}^\dag)}{\log\rho}\notag\\
&=&\lim_{\rho\rightarrow\infty}\frac{\log[1+\frac{\rho|h_{s_l,t}|^2}{M}(1+|\beta_{l,1}h_{s_1,s_l}|^2+...+|\beta_{l,l-1}h_{s_{l-1},s_l}|^2+|\beta_{l,l+1}h_{s_{l+1},s_l}|^2+...)]}{\log\rho}\notag\\
&=&(1-v_l)^+.
\label{eqn:m+lthd}
\end{eqnarray}
Thus, it is clear that the outage events set $O^+$ should be defined as
\begin{eqnarray}
O^+=\{(\textbf{v,u})\in{}\mathbb{R}^{M^2,+}|2M(1-v_l)^+<2Mr\}.
\label{eqn:outagecma}
\end{eqnarray}
Consequently, we know that in order to let the outage events happen, we must have
\begin{eqnarray}
v_l>1-r.
\label{eqn:vl}
\end{eqnarray}

Using similar techniques for CBC-SSAF, we can prove the outage events described by \eqref{eqn:outagecma} dominate the probability of error of the joint ML decoder. Thus, we can use \eqref{eqn:vl} to lower-bound $d_O(r)$ and $d(r)$ as
\begin{eqnarray}
d(r)\geqslant{}d_O(r)=\inf_{\textbf{v,u}\in{O^+}}\sum_{l=1}^M[v_l+\sum_{k=1,k\neq{l}}^Mu_{k,l}]>M(1-r).
\label{eqn:dmtcma}
\end{eqnarray}

Finally, the genie-aided MISO channel provides a DMT upper bound for every CMA so that
\begin{eqnarray}
d(r)\leqslant{}M(1-r).
\end{eqnarray}
Thus, we conclude that the DMT of our proposed CMA-SSAF strategy is exactly
\begin{eqnarray}
d(r)=M(1-r).
\end{eqnarray}
\bibliographystyle{IEEEtran}
\bibliography{IEEEtrans}
\end{document}